\pdfoutput=1
\documentclass[runningheads]{llncs}

\usepackage{microtype}
\usepackage{cite}
\usepackage{hyperref}
\usepackage{xcolor}
\usepackage{amsmath,amssymb}
\usepackage{tabularx}
\usepackage{comment}
\usepackage{myrules}
\renewcommand{\rulesep}{\qquad}

\newcommand{\ottnt}[1]{\mathit{#1}}
\newcommand{\ottmv}[1]{\mathit{#1}}
\newcommand{\ottkw}[1]{\mathbf{#1}}
\newcommand{\ottsym}[1]{#1}

\renewcommand{\ottkw}[1]{\mathtt{#1} }
\renewcommand{\ottnt}[1]{#1}
\renewcommand{\ottmv}[1]{#1}
\renewcommand{\ottsym}[1]{#1}

\spnewtheorem*{convention}{Convention}{\bfseries}{\rm}
\spnewtheorem*{notation}{Notation}{\bfseries}{\rm}

\providecommand{\msansp}{\textsf{p}}
\providecommand{\msansc}{\textsf{c}}
\providecommand{\nin}{\mathrel{\not\in}}

\providecommand{\colcoleq}{\mathrel{::=}}
\providecommand{\eqdef}{\mathrel{:=}}
\providecommand{\lAngle}{\langle\!\langle}
\providecommand{\rAngle}{\rangle\!\rangle}
\providecommand{\mscrE}{\mathcal{E}}
\newcommand{\PCFvD}{PCFv\(\Delta_\mathrm{H}\)}
\newcommand{\hatched}[1]{{\setlength{\fboxsep}{2pt}\colorbox{lightgray}{\vphantom{(}\(#1\)}}}
\newcommand{\secref}[1]{\autoref{sec:#1}}
\newcommand{\figref}[1]{\autoref{fig:#1}}
\newcommand{\propref}[1]{\autoref{prop:#1}}
\newcommand{\NOTE}[1]{}
\newcommand{\AI}[1]{}

\begin{document}
\title{Manifest Contracts with Intersection Types}
\author{Yuki Nishida
\and Atsushi Igarashi\orcidID{0000-0002-5143-9764}}
\institute{Graduate School of Informatics, Kyoto University, Kyoto, Japan \\
  \email{\{nishida,igarashi\}@fos.kuis.kyoto-u.ac.jp}}

\maketitle

\begin{abstract}
  We present a \emph{manifest contract system} \PCFvD{} with
\emph{intersection types}.  A manifest contract system is a typed
functional calculus in which software contracts are integrated into a
refinement type system and consistency of contracts is checked by
combination of compile- and run-time type checking.  Intersection
types naturally arise when a contract is expressed by a conjunction of
smaller contracts.  Run-time contract checking for conjunctive
higher-order contracts in an untyped language has been studied but our
typed setting poses an additional challenge due to the fact that an
expression of an intersection type $\tau_{{\mathrm{1}}}  \wedge  \tau_{{\mathrm{2}}}$ may have to perform
different run-time checking whether it is used as $\tau_{{\mathrm{1}}}$ or
$\tau_{{\mathrm{2}}}$.

We build \PCFvD{} on top of the \(\Delta\)-calculus, a Church-style
intersection type system by Liquori and Stolze.  In the
\(\Delta\)-calculus, a canonical expression of an intersection type is
a \emph{strong pair}, whose elements are the same expressions except
for type annotations.  To address the challenge above, we relax strong
pairs so that expressions in a pair are the same except for type
annotations and casts, which are a construct for run-time checking.

We give a formal definition of \PCFvD{} and show its basic properties
as a manifest contract system: preservation, progress, and value
inversion.  Furthermore, we show
that run-time checking does not affect essential computation.


\end{abstract}

\section{Introduction}\label{sec:introduction}
\emph{Manifest contract systems}~\cite{GreenbergPW2010, WadlerF2009,
  SekiyamaIG2017, SekiyamaI2017, SekiyamaNI2015, BeloGIP2011, Greenberg2015,
  GronskiKTFF2006, Flanagan2006, KnowlesF2010, NishidaI2018}, which are typed
functional calculi, are one discipline handling \emph{software
  contracts}~\cite{Meyer1997}.  The distinguishing feature of manifest contract
systems is that they integrate contracts into a type system and guarantee some
sort of satisfiability against contracts in a program as type soundness.
%
%
Specifically, a contract is embedded into a type by means of \emph{refinement
  types} of the form \(\ottsym{\{}   x \mathord: \tau   \mid  M  \ottsym{\}}\), which represents the subset of the
\emph{underlying type} \(\tau\) such that the values in the subset satisfy the
\emph{predicate} \(M\), which can be an arbitrary Boolean expression in the
programming language.  Using the refinement types, for example, we can express
the contract of a division function, which would say ``... the divisor shall
not be zero ...'', by the type \(\ottkw{int}  \rightarrow  \ottsym{\{}   x \mathord: \ottkw{int}   \mid   x  \neq  0   \ottsym{\}}  \rightarrow  \ottkw{int}\).  In addition to
the refinement types, manifest contract systems are often equipped with \emph{dependent
  function types} in order to express more detailed contracts.  A dependent
function type, written \(\ottsym{(}   x \mathord: \sigma   \ottsym{)}  \rightarrow  \tau\) in this paper, is a type of a function
which takes one argument of the type \(\sigma\) and returns a value of the type
\(\tau\); the distinguished point from ordinary function types is that
\(\tau\) can refer to the given argument represented by \(x\).  Hence, for
example, the type of a division function can be made more specific like
\(\ottsym{(}   x \mathord: \ottkw{int}   \ottsym{)}  \rightarrow  \ottsym{(}   y \mathord: \ottsym{\{}   x' \mathord: \ottkw{int}   \mid   x'  \neq  0   \ottsym{\}}   \ottsym{)}  \rightarrow  \ottsym{\{}   z \mathord: \ottkw{int}   \mid   x  \ottsym{=}  z  \times  y   \ottsym{\}}\).  (Here, for simplicity, we
ignore the case where devision involves a remainder, though it can be
taken account into by writing a more sophisticated predicate.)

A manifest contract system checks a contract dynamically to achieve its
goal---as many \emph{correct} programs as possible can be compiled and run;
while some studies~\cite{RondonKJ2008, UnnoK2009, KobayashiSU2011, Terauchi2010,
  ZhuJ2013, VazouSJVP2014}, which also use a refinement type system, check
contract satisfaction statically but with false positives and/or restriction on predicates.
The checks are done in the form of explicit casts of the form
\(\ottsym{(}  M  \ottsym{:}  \sigma  \Rightarrow  \tau  \ottsym{)}\); where \(M\) is a subject, \(\sigma\) is a source type
(namely the type of \(M\)), and \(\tau\) is a target type.\footnote{Many
  manifest contract systems put a unique label on each cast to distinguish which
  cast fails, but we omit them for simplicity.}
A cast checks whether the value of \(M\) can
have the type \(\tau\).  If the check fails, the cast throws an uncatchable
exception called \emph{blame}, which stands for contract violation.  So, the
system does not guarantee the absence of contract violations statically, but it
guarantees that the result of successful execution satisfies the predicate of a
refinement type in the program's type.  This property follows subject reduction
and a property called \emph{value inversion}~\cite{SekiyamaNI2015}---\textit{if
  a value \(V\) has a type \(\ottsym{\{}   x \mathord: \tau   \mid  M  \ottsym{\}}\), then the expression obtained by
  substituting \(V\) for \(x\) in \(M\) is always evaluated into
  \(\ottkw{true}\)}.

\subsection{Motivation}
The motivation of the integration of intersection types is to enrich
the expressiveness of contracts by types.  It naturally arises when we
consider a contract stated in a conjunctive form~\cite{plt-tr2,
  KeilT2015, CastagnaL2017}.  Considering parities (even/odd)
of integers, for example, we can state a contract of the addition as a conjunctive form;
that is
\begin{quote}
  ``An even integer is returned if both given arguments are even integers;
  \textbf{and} an odd integer is returned if the first given argument is even
  integer and the second given argument is odd integer; \textbf{and} ...''
\end{quote}
Using intersection types, we can write the contract as the following
type.\footnote{\(\ottkw{even} \eqdef \ottsym{\{}   x \mathord: \ottkw{nat}   \mid  x \, \ottkw{mod} \, 2  \ottsym{=}  0  \ottsym{\}} \quad \ottkw{odd} \eqdef
  \ottsym{\{}   x \mathord: \ottkw{nat}   \mid  x \, \ottkw{mod} \, 2  \ottsym{=}  1  \ottsym{\}}\)}
\begin{multline*}
  \ottsym{(}  \ottkw{even}  \rightarrow  \ottkw{even}  \rightarrow  \ottkw{even}  \ottsym{)}  \wedge  \ottsym{(}  \ottkw{even}  \rightarrow  \ottkw{odd}  \rightarrow  \ottkw{odd}  \ottsym{)} \\
  \wedge \ottsym{(}  \ottkw{odd}  \rightarrow  \ottkw{even}  \rightarrow  \ottkw{odd}  \ottsym{)}  \wedge  \ottsym{(}  \ottkw{odd}  \rightarrow  \ottkw{odd}  \rightarrow  \ottkw{even}  \ottsym{)}
\end{multline*}
In fact, a semantically equivalent contract could be expressed by using
dependent function types found in existing
systems as follows, where \(\ottkw{evenp} \eqdef \lambda   x \mathord: \ottkw{nat}   \ottsym{.}  x \, \ottkw{mod} \, 2  \ottsym{=}  0\) and
\(\ottkw{oddp} \eqdef \lambda   x \mathord: \ottkw{nat}   \ottsym{.}  x \, \ottkw{mod} \, 2  \ottsym{=}  1\).
\begin{alignat*}{1}
  ( x \mathord: \ottkw{nat} ) \rightarrow ( y \mathord: \ottkw{nat} ) \rightarrow  \{  z \mathord: \ottkw{nat}  \mid &\ottkw{if} \, \ottkw{evenp} \, x \\
  & \ottkw{then} \, \ottsym{(}  \ottkw{if} \, \ottkw{evenp} \, y \, \ottkw{then} \, \ottkw{evenp} \, z \, \ottkw{else} \, \ottkw{oddp} \, z  \ottsym{)}\\
  & \ottkw{else} \, \ottsym{(}  \ottkw{if} \, \ottkw{evenp} \, y \, \ottkw{then} \, \ottkw{oddp} \, z \, \ottkw{else} \, \ottkw{evenp} \, z  \ottsym{)}\}
\end{alignat*}
Thus, one might think it is just a matter of taste in how contracts are represented.
However, intersection types are more expressive, that is, there are contracts that
are hard to express in existing manifest contract systems.  Consider the
following (a bit contrived) contract for a higher-order function.
\begin{gather*}
  \ottsym{(}  \ottsym{(}  \ottkw{int}  \rightarrow  \ottsym{\{}   x \mathord: \ottkw{int}   \mid   x  \neq  0   \ottsym{\}}  \ottsym{)}  \rightarrow  \ottsym{\{}   z \mathord: \ottkw{int}   \mid  z  \ottsym{=}  1  \ottsym{\}}  \ottsym{)}  \wedge  \ottsym{(}  \ottsym{(}  \ottkw{int}  \rightarrow  \ottkw{int}  \ottsym{)}  \rightarrow  \ottsym{\{}   z \mathord: \ottkw{int}   \mid  z  \ottsym{=}  0  \ottsym{\}}  \ottsym{)}
\end{gather*}
The result type depends on input as the parity contract does.  This
time, however, it cannot be written with a dependent function type; there
is no obvious way to write a predicate corresponding to \(\ottkw{evenp}\) (or
\(\ottkw{oddp}\)).  Such a predicate must check that a given function returns
non-zero for all integers, but this is simply not computable.

\subsection{Our Work}
We develop a formal calculus \PCFvD{}, a manifest contract system with
intersection types.  The goal of this paper is to prove its desirable
properties: preservation, progress, value inversion; and one that guarantees
that the existence of dynamic checking does not change the ``essence'' of
computation.

There are several tasks in constructing a manifest contract system, but a
specific challenge for \PCFvD{} arises from the fact---manifest contract systems
are intended as an intermediate language for \emph{hybrid type
  checking}~\cite{Flanagan2006}.  Firstly, consider the following definition
with a parity contract in a surface language.
\begin{gather*}
  \ottkw{let} \,  succ' \mathord: \ottkw{odd}  \rightarrow  \ottkw{even}   \ottsym{=}  \lambda  x  \ottsym{.}   \ottkw{succ} ( x ) .
\end{gather*}
Supposing the primitive operator \( \ottkw{succ} ( x ) \) has the type
\(\ottkw{nat}  \rightarrow  \ottkw{nat}\), we need to check subtyping relation \(\ottkw{odd}  \ottsym{<:}  \ottkw{nat}\) and
\(\ottkw{nat}  \ottsym{<:}  \ottkw{even}\) to check well-typedness of the definition.  As we have
mentioned, however, this kind of subtyping checking is undecidable in general.
So, (when the checking is impossible) we insert casts to check the contract at
run-time and obtain the following compiled definition.
\begin{gather*}
  \ottkw{let} \,  succ' \mathord: \ottkw{odd}  \rightarrow  \ottkw{even}   \ottsym{=}  \lambda   x \mathord: \ottkw{odd}   \ottsym{.}  \ottsym{(}   \ottkw{succ} ( \ottsym{(}  x  \ottsym{:}  \ottkw{odd}  \Rightarrow  \ottkw{nat}  \ottsym{)} )   \ottsym{:}  \ottkw{nat}  \Rightarrow  \ottkw{even}  \ottsym{)}.
\end{gather*}
A problem arises when we consider the following definition equipped with a more
complicated parity contract.
\begin{gather*}
  \ottkw{let} \,  succ' \mathord: \ottsym{(}  \ottkw{odd}  \rightarrow  \ottkw{even}  \ottsym{)}  \wedge  \ottsym{(}  \ottkw{even}  \rightarrow  \ottkw{odd}  \ottsym{)}   \ottsym{=}  \lambda  x  \ottsym{.}   \ottkw{succ} ( x ) .
\end{gather*}
%
The problem is that we need to insert different casts into code according to how
the code is typed; and one piece of code might be typed in several essentially
different ways in an intersection type system since it is a polymorphic type
system.  For instance, in the example above,
\(\lambda   x \mathord: \ottkw{odd}   \ottsym{.}  \ottsym{(}   \ottkw{succ} ( \ottsym{(}  x  \ottsym{:}  \ottkw{odd}  \Rightarrow  \ottkw{nat}  \ottsym{)} )   \ottsym{:}  \ottkw{nat}  \Rightarrow  \ottkw{even}  \ottsym{)}\) is obtained by cast insertion if
the function is typed as \(\ottkw{odd}  \rightarrow  \ottkw{even}\); while
\(\lambda   x \mathord: \ottkw{even}   \ottsym{.}  \ottsym{(}   \ottkw{succ} ( \ottsym{(}  x  \ottsym{:}  \ottkw{even}  \Rightarrow  \ottkw{nat}  \ottsym{)} )   \ottsym{:}  \ottkw{nat}  \Rightarrow  \ottkw{odd}  \ottsym{)}\) is obtained when the body is
typed as \(\ottkw{even}  \rightarrow  \ottkw{odd}\).  However, the function must have both types to have
the intersection type.  It may seem sufficient to just cast the body itself,
that is, \(\ottsym{(}  \ottsym{(}  \lambda   x \mathord: \ottkw{nat}   \ottsym{.}   \ottkw{succ} ( x )   \ottsym{)}  \ottsym{:}  \ottkw{nat}  \rightarrow  \ottkw{nat}  \Rightarrow  \ottsym{(}  \ottkw{odd}  \rightarrow  \ottkw{even}  \ottsym{)}  \wedge  \ottsym{(}  \ottkw{even}  \rightarrow  \ottkw{odd}  \ottsym{)}  \ottsym{)}\).
However, this just shelves the problem: Intuitively, to check if the subject has
the target intersection type, we need to check if the subject has both types in
the conjunction.  This brings us back to the same original question.

\subsubsection{Contributions.}
Our contributions are summarized as follows:
\begin{itemize}
\item we design a manifest contracts calculus with \emph{refinement intersection
    types}~\cite{Terauchi2010, ZhuJ2013}, a restricted form of intersection
  types.
\item we formalize the calculus \PCFvD{}; and 
\item we state and prove type soundness, value inversion, and dynamic soundness.
\end{itemize}
The whole system including proofs is mechanized with Coq.\footnote{%
  The Coq scripts are available through the following URL:
  \url{https://www.fos.kuis.kyoto-u.ac.jp/~igarashi/papers/manifest-intersection.html}.}
We use locally nameless representation and cofinite
quantification~\cite{Chargueraud2012} for the mechanization.

\subsubsection{Disclaimer.}
To concentrate on the \PCFvD{}-specific problems, we put the following
restrictions for \PCFvD{} in this paper compared to a system one would imagine
from the phrase ``a manifest contract system with intersection types''.
\begin{itemize}
\item \PCFvD{} does not support dependent function types.  As we will see,
  \PCFvD{} uses nondeterminism for dynamic checking.  The combination of
  dependent function types and nondeterminism poses a considerable
  challenge~\cite{NishidaI2018}.
\item We use \emph{refinement intersection types} rather than general ones.  Roughly
  speaking, \(\sigma  \wedge  \tau\) is a refinement intersection type if both \(\sigma\)
  and \(\tau\) refine the same type.  So, for example,
  \(\ottsym{(}  \ottkw{even}  \rightarrow  \ottkw{even}  \ottsym{)}  \wedge  \ottsym{(}  \ottkw{odd}  \rightarrow  \ottkw{odd}  \ottsym{)}\) is a refinement intersection types since
  types of both sides refine the same type \(\ottkw{nat}  \rightarrow  \ottkw{nat}\), while
  \(\ottsym{(}  \ottkw{nat}  \rightarrow  \ottkw{nat}  \ottsym{)}  \wedge  \ottsym{(}  \ottkw{float}  \rightarrow  \ottkw{float}  \ottsym{)}\) is not.
\end{itemize}


\section{Overview of Our Language: \PCFvD{}}\label{sec:language}
Our language \PCFvD{} is a call-by-value dialect of
PCF~\cite{Plotkin1977}, extended with intersection types (derived from
the \(\Delta\)-calculus~\cite{LiquoriS2018}) and manifest contracts
(derived from
\(\lambda_{\textrm{H}}\)~\cite{Flanagan2006,GreenbergPW2010}).  So,
the baseline is that any \emph{valid} PCF program is also a valid
\PCFvD{} program; and a \PCFvD{} program should behave as the same way
as (call-by-value) PCF.  In other words, \PCFvD{} is a conservative
extension of call-by-value PCF.

\subsection{The \(\Delta\)-calculus}

To address the challenge discussed in \secref{introduction}, \PCFvD{} is strongly
influenced by the \emph{\(\Delta\)-calculus} by Liquori and Stolze~\cite{LiquoriS2018}, an intersection type
system \`a la Church.  Their novel idea is a new form called \emph{strong
  pair}, written \(\langle  M  \ottsym{,}  N  \rangle\).  It is a kind of pair and used as a constructor
for expressions of intersection types.  So, using the strong pair, for example,
we can write an identity function having type
\(\ottsym{(}  \ottkw{even}  \rightarrow  \ottkw{even}  \ottsym{)}  \wedge  \ottsym{(}  \ottkw{odd}  \rightarrow  \ottkw{odd}  \ottsym{)}\) as follows.
\begin{gather*}
  \langle  \lambda   x \mathord: \ottkw{even}   \ottsym{.}  x  \ottsym{,}  \lambda   x \mathord: \ottkw{odd}   \ottsym{.}  x  \rangle
\end{gather*}
Unlike product types, however, $M$ and $N$ in a strong pair cannot be arbitrarily chosen.  A strong pair requires
%
%
that the \emph{essence} of both expressions in a pair be the same.  An
essence \( \mathopen\wr M \mathclose\wr \) of a typed expression $M$ is the untyped skeleton
of $M$.  For instance, \(  \mathopen\wr \lambda   x \mathord: \tau   \ottsym{.}  x \mathclose\wr  \ottsym{=} \lambda  x  \ottsym{.}  x \).  So, the
requirement justifies strong pairs as the introduction of intersection
types: that is, computation represented by the two expressions is the same
and so the system still follows a Curry-style intersection type
system.  Strong pairs just give a way to annotate expressions with a
different type in a different context.

We adapt their idea into \PCFvD{} by letting an essence represent the
\emph{contract-irrelevant part} of an expression, rather than an untyped
skeleton.  For instance, the essence of
\(\lambda   x \mathord: \ottkw{odd}   \ottsym{.}  \ottsym{(}   \ottkw{succ} ( \ottsym{(}  x  \ottsym{:}  \ottkw{odd}  \Rightarrow  \ottkw{nat}  \ottsym{)} )   \ottsym{:}  \ottkw{nat}  \Rightarrow  \ottkw{even}  \ottsym{)}\) is \(\lambda   x \mathord: \ottkw{nat}   \ottsym{.}   \ottkw{succ} ( x ) \) (the
erased contract-relevant parts are casts and predicates of refinement types).
Now, we can (ideally automatically) compile the \( succ' \) definition in
\secref{introduction} into the following \PCFvD{} expression.
\begin{multline*}
  \ottkw{let} \,  succ' \mathord: \ottsym{(}  \ottkw{odd}  \rightarrow  \ottkw{even}  \ottsym{)}  \wedge  \ottsym{(}  \ottkw{even}  \rightarrow  \ottkw{odd}  \ottsym{)}   \ottsym{=}   \\ 
  \langle\lambda   x \mathord: \ottkw{odd}   \ottsym{.}  \ottsym{(}   \ottkw{succ} ( \ottsym{(}  x  \ottsym{:}  \ottkw{odd}  \Rightarrow  \ottkw{nat}  \ottsym{)} )   \ottsym{:}  \ottkw{nat}  \Rightarrow  \ottkw{even}  \ottsym{)},\\
  \lambda   x \mathord: \ottkw{even}   \ottsym{.}  \ottsym{(}   \ottkw{succ} ( \ottsym{(}  x  \ottsym{:}  \ottkw{even}  \Rightarrow  \ottkw{nat}  \ottsym{)} )   \ottsym{:}  \ottkw{nat}  \Rightarrow  \ottkw{odd}  \ottsym{)}\rangle
\end{multline*}
This strong pair satisfies the condition, that is, both expressions
have the same essence.

\subsection{Cast Semantics for Intersection Types}
Having introduced intersection types, we have to extend the semantics
of casts so that they handle contracts written with intersection types.
Following Keil and Thiemann~\cite{KeilT2015}, who studied intersection
(and union) contract checking in the ``latent'' style~\cite{GreenbergPW2010}
for an untyped language, we give the semantics of 
a cast \emph{to} an intersection type by the following rule:
\begin{gather*}
  \ottsym{(}  V  \ottsym{:}  \sigma  \Rightarrow  \tau_{{\mathrm{1}}}  \wedge  \tau_{{\mathrm{2}}}  \ottsym{)}  \longrightarrow  \langle  \ottsym{(}  V  \ottsym{:}  \sigma  \Rightarrow  \tau_{{\mathrm{1}}}  \ottsym{)}  \ottsym{,}  \ottsym{(}  V  \ottsym{:}  \sigma  \Rightarrow  \tau_{{\mathrm{2}}}  \ottsym{)}  \rangle
\end{gather*}
The reduction rule should not be surprising: \(V\) has to have both
\(\tau_{{\mathrm{1}}}\) and \(\tau_{{\mathrm{2}}}\) and a strong pair introduces an intersection type
$\tau_{{\mathrm{1}}}  \wedge  \tau_{{\mathrm{2}}}$ from $\tau_{{\mathrm{1}}}$ and $\tau_{{\mathrm{2}}}$.  For the original cast to succeed, both
of the split casts have to succeed.


A basic strategy of a cast \emph{from} an intersection type is expressed by
the following two rules.
\begin{gather*}
  \ottsym{(}  V  \ottsym{:}  \sigma_{{\mathrm{1}}}  \wedge  \sigma_{{\mathrm{2}}}  \Rightarrow  \tau  \ottsym{)}  \longrightarrow  \ottsym{(}   \pi_\mathrm{1} ( V )   \ottsym{:}  \sigma_{{\mathrm{1}}}  \Rightarrow  \tau  \ottsym{)}\\
  \ottsym{(}  V  \ottsym{:}  \sigma_{{\mathrm{1}}}  \wedge  \sigma_{{\mathrm{2}}}  \Rightarrow  \tau  \ottsym{)}  \longrightarrow  \ottsym{(}   \pi_\mathrm{2} ( V )   \ottsym{:}  \sigma_{{\mathrm{2}}}  \Rightarrow  \tau  \ottsym{)}
\end{gather*}
The cast tests whether a nondeterministically chosen element in a
(possibly nested) strong pair can be cast to $\tau$.

One problem, however, arises when a function type is involved.
Consider the following expression.
\begin{gather*}
    \ottsym{(}   \lambda   f \mathord: \ottkw{nat}  \rightarrow  \ottkw{nat}   \ottsym{.}  f \, 0  +  f  \, 1  \ottsym{)} \, M_{\text{cast}}
  \end{gather*}
  where
  \begin{gather*}
  M_{\text{cast}} \eqdef \ottsym{(}  V  \ottsym{:}  \ottsym{(}  \ottkw{even}  \rightarrow  \ottkw{nat}  \ottsym{)}  \wedge  \ottsym{(}  \ottkw{odd}  \rightarrow  \ottkw{nat}  \ottsym{)}  \Rightarrow  \ottkw{nat}  \rightarrow  \ottkw{nat}  \ottsym{)}.
\end{gather*}
\(V\) can be used as both \(\ottkw{even}  \rightarrow  \ottkw{nat}\) and \(\ottkw{odd}  \rightarrow  \ottkw{nat}\).
This means \(V\) can handle arbitrary natural numbers.  Thus, this
cast should be valid and evaluation of the expression above should not fail.
However, with the reduction rules presented above, evaluation results
in blame in both branches: the choice is made before calling
$\lambda f:\ottkw{nat}  \rightarrow  \ottkw{nat}. \cdots$, the function being assigned into
\(f\) only can handle either \(\ottkw{even}\) or \(\ottkw{odd}\), leading
to failure at either \(f \, 1\) or \(f \, 0\), respectively.

To solve the problem, we delay a cast into a function type even when the source type is an intersection type.  In fact, \(M_{\text{cast}}\) reduces to
a wrapped value $V_{\text{cast}}$ below
\begin{gather*}
  V_{\text{cast}} \eqdef \lAngle  V  \ottsym{:}  \ottsym{(}  \ottkw{even}  \rightarrow  \ottkw{nat}  \ottsym{)}  \wedge  \ottsym{(}  \ottkw{odd}  \rightarrow  \ottkw{nat}  \ottsym{)}  \Rightarrow  \ottkw{nat}  \rightarrow  \ottkw{nat}  \rAngle,
\end{gather*}
similarly to higher-order casts~\cite{FindlerF2002}.
Then, the delayed cast fires when an actual argument is given:
\begin{alignat*}{2}
  &&&\ottsym{(}   \lambda   f \mathord: \ottkw{nat}  \rightarrow  \ottkw{nat}   \ottsym{.}  f \, 0  +  f  \, 1  \ottsym{)} \, M_{\text{cast}}\\
  & \longrightarrow && \ottsym{(}   \lambda   f \mathord: \ottkw{nat}  \rightarrow  \ottkw{nat}   \ottsym{.}  f \, 0  +  f  \, 1  \ottsym{)} \, V_{\text{cast}}\\
  & \longrightarrow && V_{\text{cast}} \, 0 + V_{\text{cast}} \, 1\\
  & \longrightarrow^\ast \;&& \ottsym{(}  V  \ottsym{:}  \ottkw{even}  \rightarrow  \ottkw{nat}  \Rightarrow  \ottkw{nat}  \rightarrow  \ottkw{nat}  \ottsym{)} \, 0 + \ottsym{(}  V  \ottsym{:}  \ottkw{odd}  \rightarrow  \ottkw{nat}  \Rightarrow  \ottkw{nat}  \rightarrow  \ottkw{nat}  \ottsym{)}
  \, 1\\
  & \longrightarrow^\ast && 1
\end{alignat*}



\section{Formal Systems}\label{sec:system}
In this section, we formally define two languages PCFv and \PCFvD{}, an
extension of PCFv as sketched in the last section.  PCFv is a call-by-value PCF.
We only give operational semantics and omit its type system and a type soundness
proof, because we are only interested in how its behavior is related to
\PCFvD{}, the main language of this paper.

\subsection{PCFv}

\begin{figure}[t]
  \centering
  \begin{tabularx}{\linewidth}{@{}>{$}r<{\;$}@{}>{$}X<{$}@{}}
  \sigma, \tau \colcoleq & \ottkw{nat} \mid \ottkw{bool} \mid \sigma  \rightarrow  \tau
  \\
  L, M, N \colcoleq & \ottkw{O} \mid  \ottkw{succ} ( M )  \mid  \ottkw{pred} ( M )  \mid
   \ottkw{iszero} ( M )  \mid \ottkw{true} \mid \ottkw{false} \mid \ottkw{if} \, L \, \ottkw{then} \, M \, \ottkw{else} \, N \mid
  x \mid M \, N \mid \lambda   x \mathord: \tau   \ottsym{.}  M \mid  \mu  f \mathord: \sigma_{{\mathrm{1}}}  \rightarrow  \sigma_{{\mathrm{2}}}  . \lambda   x \mathord: \tau   \ottsym{.}  M 
  \\
  \overline{n} \colcoleq & \ottkw{O} \mid  \ottkw{succ} ( \overline{n} ) 
  \\
  V \colcoleq & \overline{n} \mid \ottkw{true} \mid \ottkw{false} \mid \lambda   x \mathord: \tau   \ottsym{.}  M
  \\
  \mscrE \colcoleq &  \ottkw{succ} (  \Box  )  \mid  \ottkw{pred} (  \Box  )  \mid  \ottkw{iszero} (  \Box  )  \mid
  \ottkw{if} \,  \Box  \, \ottkw{then} \, M \, \ottkw{else} \, N \mid  \Box  \, M \mid V \,  \Box 
\end{tabularx}


  \caption{Syntax of PCFv.}
  \label{fig:pcfv-syntax}
\end{figure}

The syntax of PCFv is shown in \figref{pcfv-syntax}.  Metavariables \(x\),
\(y\), \(z\), \(f\), and \(g\) range over term variables
(\(f\) and \(g\) are intended for ones bound to functions); \(\sigma\)
and \(\tau\) range over types; \(L\), \(M\), and \(N\) range over
expressions; \(V\) ranges over values; and \(\mscrE\) ranges over evaluation
frames.  The definition is fairly standard, except for one point:
instead of introducing a constant for the general \texttt{fix}-point operator,
we introduce a form $ \mu  f \mathord: \sigma_{{\mathrm{1}}}  \rightarrow  \sigma_{{\mathrm{2}}}  . \lambda   x \mathord: \tau   \ottsym{.}  M $ for recursive functions.

\begin{definition}[Bound and free variables]
  An occurrence of \(x\) in \(M\) of \(\lambda   x \mathord: \tau   \ottsym{.}  M\) and \(f\) in
  \(M\) of \( \mu  f \mathord: \sigma_{{\mathrm{1}}}  \rightarrow  \sigma_{{\mathrm{2}}}  . \lambda   x \mathord: \tau   \ottsym{.}  M \) is called \emph{bound}.
  The set of \emph{free variables} in \(M\) is the
  variables of which there are free occurrence in \(M\).  We denote the free
  variables by \( \ottkw{fv} ( M ) \).
\end{definition}

\begin{convention}
  We define \(\alpha\)-equivalence in a standard manner and identify \(\alpha\)-equivalent expressions.
\end{convention}

\begin{definition}[Substitution]
  Substitution of \(N\) for a free variable \(x\) in
  \(M\), written \( M  [  x  \mapsto  N  ] \), is defined 
  in a standard capture-avoiding manner.
\end{definition}

\begin{definition}[Context application]
  Given an evaluation frame \(\mscrE\) and an expression \(M\), \( \mscrE  [  M  ] \)
  denotes the expression obtained by just replacing the hole \( \Box \) in
  \(\mscrE\) with \(M\).
\end{definition}

\begin{figure}[t]
  \centering
  \input{figures/pcfv-semantics}
  \caption{Operational semantics of PCFv.}
  \label{fig:pcfv-semantics}
\end{figure}

A small-step operational semantics of PCFv is inductively defined by the rules
in \figref{pcfv-semantics}.  Those rules consist of standard (call-by-value) PCF
axiom schemes and one rule scheme \ruleref{PCF-Ctx}, which expresses the call-by-value evaluation strategy using the evaluation frames.

\subsection{\PCFvD{}}

\PCFvD{} is an extension of PCFv.  Through abuse of syntax, we use the
metavariables of PCFv for \PCFvD{}, though we are dealing with the two different
languages.

\begin{figure}[t]
  \centering
  \begin{tabularx}{\linewidth}{@{}>{$}r<{\;$}@{}>{$}X<{$}@{}}
  \sigma, \tau \colcoleq & \ottkw{nat} \mid \ottkw{bool} \mid \sigma  \rightarrow  \tau \mid
  \hatched{\sigma  \wedge  \tau} \mid \hatched{\ottsym{\{}   x \mathord: \tau   \mid  M  \ottsym{\}}}
  \\
  I \colcoleq & \sigma  \rightarrow  \tau \mid I_{{\mathrm{1}}}  \wedge  I_{{\mathrm{2}}}
  \\
  L, M, N \colcoleq & \ottkw{O} \mid  \ottkw{succ} ( M )  \mid  \ottkw{pred} ( M )  \mid
   \ottkw{iszero} ( M )  \mid \ottkw{true} \mid \ottkw{false} \mid \ottkw{if} \, L \, \ottkw{then} \, M \, \ottkw{else} \, N \mid
  x \mid M \, N \mid \lambda   x \mathord: \tau   \ottsym{.}  M \mid \hatched{ \mu  f \mathord: I  . B } \mid
  \hatched{\langle  M  \ottsym{,}  N  \rangle} \mid \hatched{ \pi_\mathrm{1} ( M ) } \mid \hatched{ \pi_\mathrm{2} ( M ) } \mid
  \hatched{\ottsym{(}  M  \ottsym{:}  \sigma  \Rightarrow  \tau  \ottsym{)}} \mid \hatched{\lAngle  V  \ottsym{:}  \sigma  \Rightarrow  \tau_{{\mathrm{1}}}  \rightarrow  \tau_{{\mathrm{2}}}  \rAngle} \mid
  \hatched{\lAngle  M  \mathrel?  \ottsym{\{}   x \mathord: \tau   \mid  N  \ottsym{\}}  \rAngle} \mid \hatched{\lAngle  M  \Longrightarrow  V  \ottsym{:}  \ottsym{\{}   x \mathord: \tau   \mid  N  \ottsym{\}}  \rAngle}
  \\
  B \colcoleq & \lambda   x \mathord: \tau   \ottsym{.}  M \mid \langle  B_{{\mathrm{1}}}  \ottsym{,}  B_{{\mathrm{2}}}  \rangle
  \\
  \overline{n} \colcoleq & \ottkw{O} \mid  \ottkw{succ} ( \overline{n} ) 
  \\
  V \colcoleq & \overline{n} \mid \ottkw{true} \mid \ottkw{false} \mid \lambda   x \mathord: \tau   \ottsym{.}  M \mid
  \hatched{\langle  V_{{\mathrm{1}}}  \ottsym{,}  V_{{\mathrm{2}}}  \rangle} \mid \hatched{\lAngle  V  \ottsym{:}  \sigma  \Rightarrow  \tau_{{\mathrm{1}}}  \rightarrow  \tau_{{\mathrm{2}}}  \rAngle}
  \\
  \ottnt{C} \colcoleq & M \mid \ottkw{blame}
  \\
  \mscrE \colcoleq &  \ottkw{succ} (  \Box  )  \mid  \ottkw{pred} (  \Box  )  \mid  \ottkw{iszero} (  \Box  )  \mid
  \ottkw{if} \,  \Box  \, \ottkw{then} \, M \, \ottkw{else} \, N \mid  \Box  \, M \mid V \,  \Box  \mid
  \hatched{\pi_\mathrm{1} \, \ottsym{(}   \Box   \ottsym{)}} \mid \hatched{\pi_\mathrm{2} \, \ottsym{(}   \Box   \ottsym{)}} \mid \hatched{\ottsym{(}   \Box   \ottsym{:}  \sigma  \Rightarrow  \tau  \ottsym{)}} \mid \hatched{\lAngle   \Box   \mathrel?  \ottsym{\{}   x \mathord: \tau   \mid  M  \ottsym{\}}  \rAngle}
  \\
  \Gamma \colcoleq &  \emptyset  \mid \Gamma  \ottsym{,}   x \mathord: \tau 
\end{tabularx}


  \caption{Syntax of \PCFvD{}.}
  \label{fig:pcfvd-syntax}
\end{figure}

The syntax of \PCFvD{} is shown in \figref{pcfvd-syntax}.  We
introduce some more metavariables: \(I\) ranges over
\emph{interface types}, a subset of types; \(B\) ranges over
\emph{recursion bodies}, a subset of expressions; \(\ottnt{C}\) ranges
over \emph{commands}; and \(\Gamma\) ranges over typing contexts.
Shaded parts show differences (extensions and modifications) from
PCFv.
Types are extended with intersection types and refinement types; the
restriction that a well-formed intersection type is a refinement
intersection type is enforced by the type system.
The variable \(x\) in \(N\) of \(\ottsym{\{}   x \mathord: \tau   \mid  N  \ottsym{\}}\) is bound.
An interface type,
which is a single function type or (possibly nested) intersection over function types,
is used for the type annotation for a recursive function.
Expressions are extended with ones for: strong pairs (namely, pair
construction, left projection, and right projection); casts; and
run-time expressions of the form \(\lAngle\dots\rAngle\) that can
occur at run time for dynamic checking and not in source code.
Recursion bodies are (possibly nested strong pairs) of
\(\lambda\)-abstractions.

Run-time expressions deserve detailed explanation.  A \emph{delayed check}
\(\lAngle  V  \ottsym{:}  \sigma  \Rightarrow  \tau_{{\mathrm{1}}}  \rightarrow  \tau_{{\mathrm{2}}}  \rAngle\) denotes a delayed cast into a function type, which
is used in cases such as those discussed in \secref{introduction} for
instance. A \emph{waiting check} \(\lAngle  M  \mathrel?  \ottsym{\{}   x \mathord: \tau   \mid  N  \ottsym{\}}  \rAngle\) denotes a state waiting
for the check \(M\) against \(N\) until \(M\) is evaluated into a
value.  An \emph{active check} \(\lAngle  M  \Longrightarrow  V  \ottsym{:}  \ottsym{\{}   x \mathord: \tau   \mid  N  \ottsym{\}}  \rAngle\) is a state running
test $M$ to see if \(V\) satisfies \(N\).
The variable $x$ in $N$ of \(\lAngle  M  \mathrel?  \ottsym{\{}   x \mathord: \tau   \mid  N  \ottsym{\}}  \rAngle\)
  and \(\lAngle  M  \Longrightarrow  V  \ottsym{:}  \ottsym{\{}   x \mathord: \tau   \mid  N  \ottsym{\}}  \rAngle\) is bound.

We do not include \(\ottkw{blame}\) in expressions, although existing manifest
contract systems usually include it among expressions.  As a consequence, the
evaluation relation for \PCFvD{} is defined between commands.  This distinction
will turn out to be convenient in stating correspondence between the semantics
of \PCFvD{} and that of PCFv, which does not have \(\ottkw{blame}\).

\begin{convention}
  We assume the index variable \(i\) ranges over \(\{1,2\}\) to save space.
\end{convention}

\begin{definition}[Terms]
  We call the union of the sets of types and expressions as \emph{terms}.
\end{definition}

\begin{notation}
  \(M \preceq N\) denotes that \textit{\(M\) is a sub-expression of \(N\)}.
\end{notation}

\begin{convention}
  We define \(\alpha\)-equivalence in a standard manner and identify \(\alpha\)-equivalent terms.
\end{convention}

\begin{convention}
  We often omit the empty environment.  We abuse a comma for the concatenation
  of environments like \(\Gamma_{{\mathrm{1}}}  \ottsym{,}  \Gamma_{{\mathrm{2}}}\).  We denote a singleton environment, an
  environment that contains only one variable binding, by \( x \mathord: \tau \).
\end{convention}

\begin{definition}[Free variables and substitution]
  Free variables and substitution are defined similarly to PCFv; and
  we use the same notations.  Note that since the types and expressions
  of \PCFvD{} are mutually recursively defined, the metaoperations are
  inductively defined for terms.
\end{definition}

\begin{definition}[Domain of typing context]
  The \emph{domain} of \(\Gamma\), written \( \ottkw{dom} ( \Gamma ) \), is defined
  by: \(  \ottkw{dom} (  \emptyset  )  \ottsym{=}  \emptyset  \) and
  \(  \ottkw{dom} ( \Gamma  \ottsym{,}   x \mathord: \tau  )  \ottsym{=}  \ottkw{dom} ( \Gamma )   \cup \{x\}\).  We abbreviate \( x  \nin   \ottkw{dom} ( \Gamma )  \) to
  \(x  \mathrel{\#}  \Gamma\).
\end{definition}

\begin{figure}[t]
  \centering
  \begin{tabular}{>{$}c<{$}>{$}c<{$}}
      \mathopen\wr \ottkw{nat} \mathclose\wr  \ottsym{=} \ottkw{nat}  &
      \mathopen\wr \ottkw{if} \, L \, \ottkw{then} \, M \, \ottkw{else} \, N \mathclose\wr  \ottsym{=} \ottkw{if} \,  \mathopen\wr L \mathclose\wr  \, \ottkw{then} \,  \mathopen\wr M \mathclose\wr  \, \ottkw{else} \,  \mathopen\wr N \mathclose\wr  \\
      \mathopen\wr \ottkw{bool} \mathclose\wr  \ottsym{=} \ottkw{bool}  &
      \mathopen\wr x \mathclose\wr  \ottsym{=} x \\
      \mathopen\wr \sigma  \rightarrow  \tau \mathclose\wr  \ottsym{=}  \mathopen\wr \sigma \mathclose\wr   \rightarrow   \mathopen\wr \tau \mathclose\wr   &
      \mathopen\wr M \, N \mathclose\wr  \ottsym{=}  \mathopen\wr M \mathclose\wr  \,  \mathopen\wr N \mathclose\wr  \\
      \mathopen\wr \sigma  \wedge  \tau \mathclose\wr  \ottsym{=}  \mathopen\wr \sigma \mathclose\wr   &
      \mathopen\wr \lambda   x \mathord: \tau   \ottsym{.}  M \mathclose\wr  \ottsym{=} \lambda   x \mathord:  \mathopen\wr \tau \mathclose\wr    \ottsym{.}   \mathopen\wr M \mathclose\wr  \\
      \mathopen\wr \ottsym{\{}   x \mathord: \tau   \mid  M  \ottsym{\}} \mathclose\wr  \ottsym{=}  \mathopen\wr \tau \mathclose\wr   &
      \mathopen\wr \langle  M  \ottsym{,}  N  \rangle \mathclose\wr  \ottsym{=}  \mathopen\wr M \mathclose\wr  \\
      \mathopen\wr \ottkw{O} \mathclose\wr  \ottsym{=} \ottkw{O}  &
      \mathopen\wr  \pi_i( M )  \mathclose\wr  \ottsym{=}  \mathopen\wr M \mathclose\wr  \\
      \mathopen\wr  \ottkw{succ} ( M )  \mathclose\wr  \ottsym{=}  \ottkw{succ} (  \mathopen\wr M \mathclose\wr  )   &
      \mathopen\wr  \mu  f \mathord: I  . B  \mathclose\wr  \ottsym{=}  \mu  f \mathord:  \mathopen\wr I \mathclose\wr   .  \mathopen\wr B \mathclose\wr   \\
      \mathopen\wr  \ottkw{pred} ( M )  \mathclose\wr  \ottsym{=}  \ottkw{pred} (  \mathopen\wr M \mathclose\wr  )   &
      \mathopen\wr \ottsym{(}  M  \ottsym{:}  \sigma  \Rightarrow  \tau  \ottsym{)} \mathclose\wr  \ottsym{=}  \mathopen\wr M \mathclose\wr  \\
      \mathopen\wr  \ottkw{iszero} ( M )  \mathclose\wr  \ottsym{=}  \ottkw{iszero} (  \mathopen\wr M \mathclose\wr  )   &
      \mathopen\wr \lAngle  V  \ottsym{:}  \sigma  \Rightarrow  \tau_{{\mathrm{1}}}  \rightarrow  \tau_{{\mathrm{2}}}  \rAngle \mathclose\wr  \ottsym{=}  \mathopen\wr V \mathclose\wr  \\
      \mathopen\wr \ottkw{true} \mathclose\wr  \ottsym{=} \ottkw{true}  &
      \mathopen\wr \lAngle  M  \mathrel?  \ottsym{\{}   x \mathord: \tau   \mid  N  \ottsym{\}}  \rAngle \mathclose\wr  \ottsym{=}  \mathopen\wr M \mathclose\wr  \\
      \mathopen\wr \ottkw{false} \mathclose\wr  \ottsym{=} \ottkw{false}  &
      \mathopen\wr \lAngle  M  \Longrightarrow  V  \ottsym{:}  \ottsym{\{}   x \mathord: \tau   \mid  N  \ottsym{\}}  \rAngle \mathclose\wr  \ottsym{=}  \mathopen\wr V \mathclose\wr  
\end{tabular}
  

  \caption{Essence of a \PCFvD{} term.}
  \label{fig:essence}
\end{figure}

The essence of a \PCFvD{} term is defined in \figref{essence}, which is mostly
straightforward.  The choice of which part we take as the essence of a strong pair
is arbitrary because for a well-typed expression both parts have the same
essence.  Note that the essence of an active check \(\lAngle  M  \Longrightarrow  V  \ottsym{:}  \ottsym{\{}   x \mathord: \tau   \mid  N  \ottsym{\}}  \rAngle\) is
\(V\) rather than \(M\).  This is because \(V\) is the subject of
the expression.

\subsection{Operational Semantics of \PCFvD{}}

The operational semantics of \PCFvD{} consists of four relations
\(M  \rightharpoonup_\msansp  N\), \(M  \rightharpoonup_\msansc  \ottnt{C}\), \(M  \longrightarrow_\msansp  N\), and \(M  \longrightarrow_\msansc  \ottnt{C}\).
Bearing in mind the inclusion relation among syntactic categories, these
relations can be regarded as binary relations between commands.  The first two
are basic reduction relations, and the other two are contextual evaluation
relations (relations for whole programs).  Furthermore, the relations
subscripted by \(\msansp\) correspond to PCFv evaluation, that is,
\emph{essential evaluation}; and ones subscripted by \(\msansc\) correspond to
dynamic contract checking.  Dynamic checking is
nondeterministic because of \ruleref{RC-WedgeL/R}, \ruleref{EC-PairL}, and
\ruleref{EC-PairR}.

\subsubsection{Essential Evaluation \( \longrightarrow_\msansp \).}

\begin{figure}[t]
  \centering
  \input{figures/pcfvd-semantics1}
  \caption{Operational semantics of \PCFvD{} (1): essential evaluation.}
  \label{fig:pcfvd-semantics1}
\end{figure}

The essential evaluation, defined in \figref{pcfvd-semantics1}, defines the
evaluation of the essential part of a program; and thus, it is similar to
\( \longrightarrow_\textsf{PCF} \).  There are just three differences, that is: there are two
relations; there is no reduction rule for \( \ottkw{pred} ( \ottkw{O} ) \); and there is a
distinguished contextual evaluation rule \ruleref{EP-PairS}, which synchronizes
essential reductions of the elements in a strong pair.  The
synchronization in \ruleref{EP-PairS} is important since a strong pair requires
the essences of both elements to be the same.  The lack of predecessor
evaluation for \(\ottkw{O}\) is intentional: Our type system and run-time checking
guarantee that \(\ottkw{O}\) cannot occur as an argument to $\mathtt{pred}$.  

\subsubsection{Dynamic Checking \( \longrightarrow_\msansc \).}

\begin{figure}[t]
  \centering
  \input{figures/pcfvd-semantics2}
  \caption{Operational semantics of \PCFvD{} (2): reduction rules for dynamic checking.}
  \label{fig:pcfvd-semantics2}
\end{figure}

Dynamic checking is more complicated.  Firstly, we focus on reduction rules in
\figref{pcfvd-semantics2}.  The side-conditions on some rules are set so that an
evaluation is less nondeterministic (for example, without the side conditions,
both \ruleref{RC-Forget} and \ruleref{RC-Delay} could be applied to one
expression).

The rules irrelevant to intersection types (\ruleref{RC-Nat},
\ruleref{RC-Bool}, \ruleref{RC-Forget}, \ruleref{RC-Delay}, \ruleref{RC-Arrow},
\ruleref{RC-Waiting}, \ruleref{RC-Activate}, \ruleref{RC-Succeed}, and
\ruleref{RC-Fail}) are adopted from Sekiyama et al.~\cite{SekiyamaNI2015}, but
there is one difference about \ruleref{RC-Delay} and \ruleref{RC-Arrow}.  In the
original definition delayed checking is done by using lambda abstractions, that
is,
\begin{gather*}
  \ottsym{(}  V  \ottsym{:}  \sigma_{{\mathrm{1}}}  \rightarrow  \sigma_{{\mathrm{2}}}  \Rightarrow  \tau_{{\mathrm{1}}}  \rightarrow  \tau_{{\mathrm{2}}}  \ottsym{)}  \longrightarrow  \lambda   x \mathord: \tau_{{\mathrm{1}}}   \ottsym{.}  \ottsym{(}  V \, \ottsym{(}  x  \ottsym{:}  \tau_{{\mathrm{1}}}  \Rightarrow  \sigma_{{\mathrm{1}}}  \ottsym{)}  \ottsym{:}  \sigma_{{\mathrm{2}}}  \Rightarrow  \tau_{{\mathrm{2}}}  \ottsym{)}.
\end{gather*}
The reason we adopt a different way is just it makes technical development
easier.  Additionally, the way we adopt is not new---It is used in the
original work~\cite{FindlerF2002} on higher-order contract calculi.

The other rules are new ones we propose for dynamic checking of intersection
types.  As we have discussed in \secref{language}, a cast into an intersection
type is reduced into a pair of casts by \ruleref{RC-WedgeI}.  A cast from an
intersection type is done by \ruleref{RC-Delay}, \ruleref{RC-WedgeL/R} if the
target type is a function type.  Otherwise, if the target type is a first order
type, \ruleref{RC-WedgeN} and \ruleref{RC-WedgeB} are used, where we arbitrarily
choose the left side of the intersection type and the corresponding part of the
value since the source type is not used for dynamic checking of first-order
values.

\begin{figure}[t]
  \centering
  \input{figures/pcfvd-semantics3}
  \caption{Operational semantics of \PCFvD{} (3): contextual rules for dynamic checking.}
  \label{fig:pcfvd-semantics3}
\end{figure}

The contextual evaluation rules, defined in \figref{pcfvd-semantics3}, are
rather straightforward.  Be aware of the use of metavariables, for instance, the
use of \(N\) in \ruleref{EC-Ctx}; it implicitly means that \(M\) has not
been evaluated into \(\ottkw{blame}\) (so the rule does not overlap with
\ruleref{EB-Ctx}).  The first rule lifts the reduction relation to the
evaluation relation.  The next six rules express the case where a sub-expression
is successfully evaluated.  The rules \ruleref{EC-ActiveP} and
\ruleref{EC-ActiveC} mean that evaluation inside an active check is always
considered dynamic checking, even when it involves essential evaluation.  The rules
\ruleref{EC-PairL} and \ruleref{EC-PairR} mean that dynamic checking does not
synchronize because the elements in a strong pair may have different casts.  The
other rules express the case where dynamic checking has failed.  An expression
evaluates to \(\ottkw{blame}\) immediately---in one step---when a sub-expression
evaluates to \(\ottkw{blame}\).  Here is an example of execution of failing dynamic
checking.
\begin{alignat*}{1}
   \ottsym{(}  0  \ottsym{:}  \ottkw{nat}  \Rightarrow  \ottsym{\{}   x \mathord: \ottkw{nat}   \mid   x  >  0   \ottsym{\}}  \ottsym{)}  +  1 
   \longrightarrow \;& \lAngle  0  \mathrel?  \ottsym{\{}   x \mathord: \ottkw{nat}   \mid   x  >  0   \ottsym{\}}  \rAngle  +  1  \\
   \longrightarrow \;& \lAngle   0  >  0   \Longrightarrow  0  \ottsym{:}  \ottsym{\{}   x \mathord: \ottkw{nat}   \mid   x  >  0   \ottsym{\}}  \rAngle  +  1  \\
   \longrightarrow \;& \lAngle  \ottkw{false}  \Longrightarrow  0  \ottsym{:}  \ottsym{\{}   x \mathord: \ottkw{nat}   \mid   x  >  0   \ottsym{\}}  \rAngle  +  1  \\
   \longrightarrow \;&\ottkw{blame}
\end{alignat*}

\begin{definition}[Evaluation]
  The one-step evaluation relation of \PCFvD{}, denoted by \( \longrightarrow \), is
  defined as \( \longrightarrow_\msansp  \cup  \longrightarrow_\msansc \).  The multi-step evaluation relation
  of \PCFvD{}, denoted by \( \longrightarrow^\ast \), is the reflexive and transitive closure
  of \( \longrightarrow \).
\end{definition}

\subsection{Type System of \PCFvD{}}

\begin{figure}[t]
  \centering
  \input{figures/typing1}
  \caption{Type system of \PCFvD{} (1): well-formedness rules.}
  \label{fig:typing1}
\end{figure}

\begin{figure}[t]
  \centering
  \input{figures/typing2}
  \caption{Type system of \PCFvD{} (2): compile-time typing rules.}
  \label{fig:typing2}
\end{figure}

\begin{figure}[t]
  \centering
  \input{figures/typing3}
  \caption{Type system of \PCFvD{} (3): run-time typing rules.}
  \label{fig:typing3}
\end{figure}

The type system consists of three judgments: \( \Gamma \;\mathrm{ok} \), \( \Vdash  \tau \),
and \(\Gamma  \vdash  M  \ottsym{:}  \tau\), read ``\(\Gamma\) is well-formed'', ``\(\tau\) is
well-formed'', and ``\(M\) has \(\tau\) under \(\Gamma\),'' respectively.
They are defined inductively by the rules in Figures~\ref{fig:typing1},
\ref{fig:typing2} and \ref{fig:typing3}.

The rules for well-formed types check that an intersection type is restricted to
a refinement intersection type by the side condition \(  \mathopen\wr \sigma \mathclose\wr  \ottsym{=}  \mathopen\wr \tau \mathclose\wr  \) in
\ruleref{W-Wedge} and that the predicate in a refinement type is a Boolean
expression by \ruleref{W-Refine}.  Note that, since \PCFvD{} has no dependent
function type, all types are closed and the predicate of a refinement type only
depends on the parameter itself.

The typing rules, the rules for the third judgment, consist of two more
sub-categories: compile-time rules and run-time rules.  Compile-time rules are
for checking a program a programmer writes.  Run-time rules are for run-time
expressions and used to prove type soundness.  This distinction, which follows,
Belo et al.~\cite{BeloGIP2011}, is to make compile-time type checking decidable.

A large part of the compile-time rules are adapted from PCF, Sekiyama et
al.~\cite{SekiyamaNI2015}, and Liquori and Stolze~\cite{LiquoriS2018}.  Here we
explain some notable rules.  As an intersection type system, \ruleref{T-Pair},
\ruleref{T-Fst}, and \ruleref{T-Snd} stands for introduction and elimination
rules of intersection types (or we can explicitly introduce and/or eliminate
an intersection type by a cast).  The rule \ruleref{T-Pair} checks a strong pair
is composed by essentially the same expressions by \(  \mathopen\wr M \mathclose\wr  \ottsym{=}  \mathopen\wr N \mathclose\wr  \).  The
rule \ruleref{T-Pred} demands that the argument of predecessor shall not be
zero.  The premise \(  \mathopen\wr \sigma \mathclose\wr  \ottsym{=}  \mathopen\wr \tau \mathclose\wr  \) of the rule \ruleref{T-Cast} for casts
requires the essences of the source and target types to agree.  It amounts to
checking the two types \(\sigma\) and \(\tau\) are
compatible~\cite{SekiyamaNI2015}.


The run-time rules are from Sekiyama et al.~\cite{SekiyamaNI2015} with one extra
rule \ruleref{T-Delayed}.  The rule \ruleref{T-Delayed} is for a delayed
checking for function types, which restrict the source type so that it respects
the evaluation relation (there is no evaluation rule for a delayed checking in
which source type is a refinement type), and inherits the condition on the
source and target types from \ruleref{T-Cast}.  The side condition
\( N  [  x  \mapsto  V  ]   \longrightarrow^\ast  M\) on \ruleref{T-Active} is an invariant during evaluation,
that is, \(M\) is an intermediate state of the predicate checking.  This
invariant lasts until the final (successful) run-time checking state
\(\lAngle  \ottkw{true}  \Longrightarrow  V  \ottsym{:}  \ottsym{\{}   x \mathord: \tau   \mid  N  \ottsym{\}}  \rAngle\) and guarantees the checking result \(V\)
(obtained by \ruleref{RC-Succeed}) satisfies the predicate \(N\) by
\ruleref{T-Exact}.



\section{Properties}\label{sec:properties}

We start from properties of evaluation relations.  As we have mentioned,
\( \longrightarrow_\msansp \) is essential evaluation, and thus, it should simulate
\( \longrightarrow_\textsf{PCF} \); and \( \longrightarrow_\msansc \) is dynamic checking, and therefore, it should
not change the essence of the expression.  We formally state and
show these properties here.  Note that most properties require that the
expression before evaluation is well typed.  This is because the condition of
strong pairs is imposed by the type system.
\begin{lemma}\label{prop:pcfdet}
  If \(M  \longrightarrow_\textsf{PCF}  N\) and \(M  \longrightarrow_\textsf{PCF}  L\), then \( N \ottsym{=} L \).
\end{lemma}
\begin{proof}
  The proof is routine by induction on one of the given derivations. \qed{}
\end{proof}
\begin{lemma}\label{prop:psim}
  If \(    \vdash  M  \ottsym{:}  \tau\) and \(M  \longrightarrow_\msansp  N\), then \( \mathopen\wr M \mathclose\wr   \longrightarrow_\textsf{PCF}   \mathopen\wr N \mathclose\wr \).
\end{lemma}
\begin{proof}
  The proof is by induction on the given evaluation derivation. \qed{}
\end{proof}
The following corollary is required to prove the preservation property.
\begin{corollary}\label{prop:psync}
  If \(    \vdash  M  \ottsym{:}  \sigma\), \(    \vdash  N  \ottsym{:}  \tau\), \(M  \longrightarrow_\msansp  M'\), \(N  \longrightarrow_\msansp  N'\), and \(  \mathopen\wr M \mathclose\wr  \ottsym{=}  \mathopen\wr N \mathclose\wr  \); then \(  \mathopen\wr M' \mathclose\wr  \ottsym{=}  \mathopen\wr N' \mathclose\wr  \).
\end{corollary}
\begin{lemma}\label{prop:cpreserve}
  If \(    \vdash  M  \ottsym{:}  \tau\) and \(M  \longrightarrow_\msansc  N\), then \(  \mathopen\wr M \mathclose\wr  \ottsym{=}  \mathopen\wr N \mathclose\wr  \).
\end{lemma}
\begin{proof}
  The proof is by induction on the given evaluation derivation. \qed{}
\end{proof}

Now we can have the following theorem as a corollary of \propref{psim} and
\propref{cpreserve}.  It guarantees the essential computation in \PCFvD{} is the
same as the PCFv computation as far as the computation does not fail.  In other
words, run-time checking may introduce blame but otherwise does not affect the
essential computation.

\begin{theorem}
  If \(    \vdash  M  \ottsym{:}  \tau\) and \(M  \longrightarrow  N\), then \( \mathopen\wr M \mathclose\wr   \longrightarrow_\textsf{PCF}^\ast   \mathopen\wr N \mathclose\wr \).
\end{theorem}

\AI{We could have shown the manifest calculus is a conservative extension of PCFv.}

\subsection{Type Soundness}

We conclude this section with type soundness.  Firstly, we show a substitution
property; and using it, we show the preservation property.
\begin{lemma}\label{prop:subst}
  If \(\Gamma_{{\mathrm{1}}}  \ottsym{,}   x \mathord: \sigma   \ottsym{,}  \Gamma_{{\mathrm{2}}}  \vdash  M  \ottsym{:}  \tau\) and \(\Gamma_{{\mathrm{1}}}  \vdash  N  \ottsym{:}  \sigma\), then \(\Gamma_{{\mathrm{1}}}  \ottsym{,}  \Gamma_{{\mathrm{2}}}  \vdash   M  [  x  \mapsto  N  ]   \ottsym{:}  \tau\).
\end{lemma}
\begin{proof}
  The proof is by induction on the derivation for \(M\). \qed{}
\end{proof}
\begin{theorem}[Preservation]
  If \(    \vdash  M  \ottsym{:}  \tau\) and \(M  \longrightarrow  N\), then \(    \vdash  N  \ottsym{:}  \tau\).
\end{theorem}
\begin{proof}
  We prove preservation properties for each \( \longrightarrow_\msansp \) and \( \longrightarrow_\msansc \) and
  combine them.  Both proofs are done by induction on the given typing
  derivation.  For the case in which substitution happens, we use
  \propref{subst} as usual.  For the context evaluation for strong pairs, we use
  \propref{psync} and \propref{cpreserve} to guarantee the side-condition of
  strong pairs.  \qed{}
\end{proof}

Next we show the value inversion property, which guarantees a value of a
refinement type satisfies its predicate.  For \PCFvD{}, this property can be
quite easily shown since \PCFvD{} does not have dependent function types, while
previous manifest contract systems need quite complicated
reasoning~\cite{SekiyamaNI2015, SekiyamaI2017, NishidaI2018}.  The property
itself is proven by using the following two, which are for strengthening an
induction hypothesis.

\begin{definition}
  We define a relation between values and types, written \(V  \models  \tau\), by the
  following rules.
  \begin{rules}
    \infbi{
      \(V  \models  \tau\)
    }{
      \( M  [  x  \mapsto  V  ]   \longrightarrow^\ast  \ottkw{true}\)
    }{
      \(V  \models  \ottsym{\{}   x \mathord: \tau   \mid  M  \ottsym{\}}\)
    }
    \infun{
      \(\ottsym{(}   \tau  \neq  \ottsym{\{}   x \mathord: \sigma   \mid  M  \ottsym{\}}   \ottsym{)}\)
    }{
      \(V  \models  \tau\)
    }
  \end{rules}
\end{definition}

\begin{lemma}\label{prop:vigen}
  If \(    \vdash  V  \ottsym{:}  \tau\), then \(V \models \tau\).
\end{lemma}
\begin{proof}
  The proof is by induction on the given derivation. \qed{}
\end{proof}

\begin{theorem}[Value inversion]\label{prop:vi}
  If \(    \vdash  V  \ottsym{:}  \ottsym{\{}   x \mathord: \tau   \mid  M  \ottsym{\}}\), then \( M  [  x  \mapsto  V  ]   \longrightarrow^\ast  \ottkw{true}\).
\end{theorem}
\begin{proof}
  Immediate from \propref{vigen}. \qed{}
\end{proof}

\begin{remark}
  As a corollary of value inversion, it follows that a value of an
  intersection type must be a strong pair and its elements satisfy the corresponding predicate in
  the intersection type: For example, if
  \(    \vdash  \langle  V_{{\mathrm{1}}}  \ottsym{,}  V_{{\mathrm{2}}}  \rangle  \ottsym{:}  \ottsym{\{}   x \mathord: \sigma   \mid  M  \ottsym{\}}  \wedge  \ottsym{\{}   x \mathord: \tau   \mid  N  \ottsym{\}}\), then \( M  [  x  \mapsto  V_{{\mathrm{1}}}  ]   \longrightarrow^\ast  \ottkw{true}\)
  and \( N  [  x  \mapsto  V_{{\mathrm{2}}}  ]   \longrightarrow^\ast  \ottkw{true}\).  In particular, for first-order values, every
  element of the pair is same.  That means the value satisfies all contracts
  concatenated by \(\wedge\).  For example,
  \(    \vdash  V  \ottsym{:}  \ottsym{\{}   x \mathord: \ottkw{nat}   \mid  M_{{\mathrm{1}}}  \ottsym{\}} \wedge \dots \wedge \ottsym{\{}   x \mathord: \ottkw{nat}   \mid  M_{\ottmv{n}}  \ottsym{\}}\), then
  \( M_{\ottmv{k}}  [  x  \mapsto   \mathopen\wr V \mathclose\wr   ]   \longrightarrow^\ast  \ottkw{true}\) for any \(k = 1..n\).  This is what we have
  desired for a contract written by using intersection types.
\end{remark}

Lastly, the progress property also holds.  In our setting, where \( \ottkw{pred} ( M ) \)
is partial, this theorem can be proved only after \propref{vi}.

\begin{theorem}[Progress]
  If \(    \vdash  M  \ottsym{:}  \tau\), then \(M\) is a value or \(M  \longrightarrow  \ottnt{C}\) for some
  \(\ottnt{C}\).
\end{theorem}
\begin{proof}
  The proof is by induction on the given derivation.  Since the evaluation
  relation is defined as combination of \( \longrightarrow_\msansp \) and \( \longrightarrow_\msansc \), the
  proof is a bit tricky, but most cases can be proven as usual.  An interesting
  case is \ruleref{T-Pair}.  We need to guarantee that if one side of a strong
  pair is a value, another side must not be evaluated by \( \longrightarrow_\msansp \) since a
  value is in normal form.  This follows from \propref{psim} and proof by
  contradiction because the essence of a \PCFvD{} value is a PCFv value and it is
  normal form. \qed{}
\end{proof}


\section{Related Work}

Intersection types were introduced in Curry-style type assignment systems by
Coppo et al.~\cite{CoppoDV1981} and Pottinger~\cite{Pottinger1980}
independently.  In the early days, intersection types are motivated by improving
a type system to make more lambda terms typeable; one important result towards
this direction is that: \textit{a lambda term has a type iff it can be strongly
  normalized}~\cite{Pottinger1980, Valentini2001}.  Then, intersection types are introduced to programming languages to enrich the descriptive power of types~\cite{Reynolds1988, BenzakenCF2003, Dunfield2007}.

\subsubsection{Intersection Contracts for Untyped Languages.}
One of the first attempts at implementing intersection-like contracts
is found in DrRacket~\cite{plt-tr2}.  It is, however, a naive
implementation, which just enforces all contracts even for functional values, and thus the semantics of higher-order intersection contracts is rather different from ours.

Keil and Thiemann~\cite{KeilT2015} have proposed an untyped calculus
of blame assignment for a higher-order contract system with
intersection and union.  As we have mentioned, our run-time checking
semantics is strongly influenced by their work, but there are two
essential differences.  On the one hand, they do not have the problem
of varying run-time checking according to a typing context; they can
freely put contract monitors\footnote{A kind of casts in their
  language.}  where they want since it is an untyped language.  On the
other hand, their operational semantics is made rather complicated due
to blame assignment.

More recently, Williams et al.~\cite{WilliamsMW2018} have proposed more sorted out
semantics for a higher-order contract system with intersection and union.  They
have mainly reformed contract checking for intersection and union ``in a uniform way''; that is, each is handled by only one
similar and simpler rule.  As a result, their presentation becomes closer to
our semantics, though complication due to blame assignment still remains.
A similar level of complication will be expected if we extend our calculus with blame assignment.

It would be interesting to investigate the
relationship between their calculi and \PCFvD{} extended with blame
labels, following Greenberg et al.~\cite{GreenbergPW2010}.

\subsubsection{Gradual Typing with Intersection Types.} Castagna and
Lanvin~\cite{CastagnaL2017} have proposed gradual typing for set-theoretic types,
which contain intersection types, as well as union and negation.  A framework of gradual typing is so close
to manifest contract systems that there is even a study unifying
them~\cite{WadlerF2009}.  A gradual typing system translates a program into an
intermediate language that is statically typed and uses casts.  Hence, they have
the same problem---how casts should be inserted when intersection types are used.  They solve the problem by
\emph{type-case} expressions, which dynamically dispatch behavior according to the type
of a value.
However, it is not clear how type-case expressions scale
to a larger language.
In fact, the following work~\cite{CastagnaLPS2019}, an extension to parametric
polymorphism and type inference, removes (necessity of) type-case
expressions but imposes instead a restriction on functions not to have an
intersection type.  Furthermore, the solution using type-case
expressions relies on strong properties of set-theoretic types.  So,
it is an open problem if their solution can be adopted to manifest
contract systems because there is not set-theoretic type theory for
refinement types and, even worse, dependent function types.

\subsubsection{Nondeterminism for Dependently Typed Languages.}
As we have noted in \secref{introduction}, \PCFvD{} has no dependent
function types.  In fact, no other work discussed in this section
supports both dependent function contracts and intersection contracts.
To extend \PCFvD{} to dependent function types, we
have to take care of their interaction with nondeterminism, which we
studied elsewhere~\cite{NishidaI2018} for a manifest calculus
\( \lambda^{H \parallel \Phi}\) with a general nondeterministic choice
operator.

A technical challenge in combining dependent function types and
nondeterminstic choice
comes from the following standard typing rule for (dependent) function applications:
\begin{prooftree}
  \AxiomC{\(\Gamma  \vdash  M  \ottsym{:}  \ottsym{(}   x \mathord: \sigma   \ottsym{)}  \rightarrow  \tau\)}%
  \AxiomC{\(\Gamma  \vdash  N  \ottsym{:}  \sigma\)}%
  \BinaryInfC{\(\Gamma  \vdash  M \, N  \ottsym{:}   \tau  [  x  \mapsto  N  ] \)}%
\end{prooftree}
The problem is that the argument $N$, which may contain
nondeterministic choice, may be duplicated in \( \tau  [  x  \mapsto  N  ] \) and, to
keep consistency of type equivalence, choices made in each occurrence
of $N$ have to be ``synchronized.''  To control synchronization,
\( \lambda^{H \parallel \Phi}\) introduces a named choice operator so that
choice operators with the same name make synchronized choice.
However, \( \lambda^{H \parallel \Phi}\) puts burden on programmers to
avoid unintended synchronization caused by accidentally shared names.

If we incoporate the idea above to \PCFvD{}, it will be natural to put
names on casts so that necessary synchronization takes place for
choices made by \ruleref{RC-WedgeL} and \ruleref{RC-WedgeR}.  It is
not clear, however, how unintended synchronization can be avoided
systematically, without programmers' ingenuity.

\section{Conclusion}\label{sec:conclusion}

We have designed and formalized a manifest contract system \PCFvD{} with
refinement intersection types.  As a result of our formal development,
\PCFvD{} guarantees not only ordinary preservation and progress but also the
property that a value of an intersection type, which can be seen as an
enumeration of small contracts, satisfies all the contracts.

The characteristic point of our formalization is that we regard a manifest
contract system as an extension of a more basic calculus, which has no software
contract system, and investigate the relationship between the basic calculus and
the manifest contract system.  More specifically, essential computation and
dynamic checking are separated.  We believe this investigation is important for
modern manifest contract systems because those become more and more complicated
and the separation is no longer admissible at a glance.

\subsubsection{Future Work.} Obvious future work is to lift the restriction we
have mentioned in \secref{introduction}.  That aside, the subsumption-free
approach is very naive and has an obvious disadvantage, that is, it requires
run-time checking even for a cast like \(\ottsym{(}  M  \ottsym{:}  \sigma  \wedge  \tau  \Rightarrow  \sigma  \ottsym{)}\), which should
be able to checked and removed at compile time.  To address the disadvantage, some
manifest contract systems provide the property known as \emph{up-cast
  elimination}~\cite{BeloGIP2011}---\textit{a cast from subtype into supertype
  can be safely removed at compile-time}.  An interesting fact is that a
well-known up-cast (subtyping) relation for a traditional intersection type
system is defined syntactically; while a usual up-cast relation for a manifest
contract system depends on semantics.  So, focusing on only the traditional
subtyping relation, the property might be proven more easily.

Towards a practice language, our cast semantics using strong pairs and
nondeterminism needs more investigation.  For the strong pairs, it will be quite
inefficient to evaluate both sides of a strong pair independently since its essence part just computes
the same thing.  The inefficiency might be reduced by a kind of sharing
structures.  For the nondeterminism, our theoretical result gives us useful
information only for successful evaluation paths; but we have not given a way to
pick up a successful one.  One obvious way is computing every evaluation path,
but of course, it is quite inefficient.


\subsection*{Acknowledgments}

We would like to thank Peter Thiemann, John Toman, Yuya Tsuda, and
anonymous reviewers for useful comments.  This work was supported in
part by the JSPS KAKENHI Grant Number JP17H01723.

\bibliographystyle{splncs04}
\bibliography{reference}

\begin{thebibliography}{10}
\providecommand{\url}[1]{\texttt{#1}}
\providecommand{\urlprefix}{URL }
\providecommand{\doi}[1]{https://doi.org/#1}

\bibitem{BeloGIP2011}
Belo, J.F., Greenberg, M., Igarashi, A., Pierce, B.C.: Polymorphic contracts.
  In: Proc.\ of {ESOP}. pp. 18--37 (2011)

\bibitem{BenzakenCF2003}
Benzaken, V., Castagna, G., Frisch, A.: {CDuce}: an {XML}-centric
  general-purpose language. In: Proc.\ of {ICFP}. pp. 51--63 (2003)

\bibitem{CastagnaL2017}
Castagna, G., Lanvin, V.: Gradual typing with union and intersection types.
  {PACMPL}  \textbf{1}({ICFP}),  41:1--41:28 (2017)

\bibitem{CastagnaLPS2019}
Castagna, G., Lanvin, V., Petrucciani, T., Siek, J.G.: Gradual typing: A new
  perspective. Proc. ACM Program. Lang.  \textbf{3}(POPL),  16:1--16:32 (Jan
  2019)

\bibitem{Chargueraud2012}
Chargu{\'{e}}raud, A.: The locally nameless representation. J. Autom. Reasoning
   \textbf{49}(3),  363--408 (2012)

\bibitem{CoppoDV1981}
Coppo, M., Dezani{-}Ciancaglini, M., Venneri, B.: Functional characters of
  solvable terms. Math. Log. Q.  \textbf{27}(2-6),  45--58 (1981)

\bibitem{Dunfield2007}
Dunfield, J.: Refined typechecking with {Stardust}. In: Proc.\ of {PLPV}. pp.
  21--32 (2007)

\bibitem{FindlerF2002}
Findler, R.B., Felleisen, M.: Contracts for higher-order functions. In: Proc.\
  of {ICFP}. pp. 48--59 (2002)

\bibitem{plt-tr2}
Findler, R.B., PLT: {DrRacket}: Programming environment. Tech. Rep.
  PLT-TR-2010-2, PLT Design Inc. (2010), \url{https://racket-lang.org/tr2/}

\bibitem{Flanagan2006}
Flanagan, C.: Hybrid type checking. In: Proc.\ of {POPL}. pp. 245--256 (2006)

\bibitem{Greenberg2015}
Greenberg, M.: Space-efficient manifest contracts. In: Proc.\ of {POPL}. pp.
  181--194 (2015)

\bibitem{GreenbergPW2010}
Greenberg, M., Pierce, B.C., Weirich, S.: Contracts made manifest. In: Proc.\
  of {POPL}. pp. 353--364 (2010)

\bibitem{GronskiKTFF2006}
Gronski, J., Knowles, K., Tomb, A., Freund, S.N., Flanagan, C.: Sage: Hybrid
  checking for flexible specifications. In: Scheme and Functional Programming
  Workshop. pp. 93--104 (2006)

\bibitem{KeilT2015}
Keil, M., Thiemann, P.: Blame assignment for higher-order contracts with
  intersection and union. In: Proc.\ of {ICFP}. pp. 375--386 (2015)

\bibitem{KnowlesF2010}
Knowles, K., Flanagan, C.: Hybrid type checking. {ACM} Trans. Program. Lang.
  Syst.  \textbf{32}(2),  6:1--6:34 (2010)

\bibitem{KobayashiSU2011}
Kobayashi, N., Sato, R., Unno, H.: Predicate abstraction and {CEGAR} for
  higher-order model checking. In: Proc.\ of {PLDI}. pp. 222--233 (2011)

\bibitem{LiquoriS2018}
Liquori, L., Stolze, C.: The {\(\Delta\)-calculus}: Syntax and types. In:
  Proc.\ of {FSCD}. pp. 28:1--28:20 (2018)

\bibitem{Meyer1997}
Meyer, B.: Object-Oriented Software Construction, 2nd Edition. Prentice-Hall
  (1997)

\bibitem{NishidaI2018}
Nishida, Y., Igarashi, A.: Nondeterministic manifest contracts. In: Proc.\ of
  {PPDP}. pp. 16:1--16:13 (2018)

\bibitem{Plotkin1977}
Plotkin, G.D.: {LCF} considered as a programming language. Theor. Comput. Sci.
  \textbf{5}(3),  223--255 (1977)

\bibitem{Pottinger1980}
Pottinger, G.: A type assignment for the strongly normalizabile
  \(\lambda\)-terms. To H.\ B.\ Curry, Essays in Combinatory Logic,
  Lambda-Calculus and Formalism pp. 561--577 (1980)

\bibitem{Reynolds1988}
Reynolds, J.C.: Preliminary design of the programming language {Forsythe}.
  Tech. Rep. CMU-CS-88-159, Carnegie Mellon University (1988)

\bibitem{RondonKJ2008}
Rondon, P.M., Kawaguchi, M., Jhala, R.: Liquid types. In: Proc.\ of {PLDI}. pp.
  159--169 (2008)

\bibitem{SekiyamaI2017}
Sekiyama, T., Igarashi, A.: Stateful manifest contracts. In: Proc.\ of {POPL}.
  pp. 530--544 (2017)

\bibitem{SekiyamaIG2017}
Sekiyama, T., Igarashi, A., Greenberg, M.: Polymorphic manifest contracts,
  revised and resolved. {ACM} Trans. Program. Lang. Syst.  \textbf{39}(1),
  3:1--3:36 (2017)

\bibitem{SekiyamaNI2015}
Sekiyama, T., Nishida, Y., Igarashi, A.: Manifest contracts for datatypes. In:
  Proc.\ of {POPL}. pp. 195--207 (2015)

\bibitem{Terauchi2010}
Terauchi, T.: Dependent types from counterexamples. In: Proc.\ of {POPL}. pp.
  119--130 (2010)

\bibitem{UnnoK2009}
Unno, H., Kobayashi, N.: Dependent type inference with interpolants. In: Proc.\
  of {PPDP}. pp. 277--288 (2009)

\bibitem{Valentini2001}
Valentini, S.: An elementary proof of strong normalization for intersection
  types. Arch. Math. Log.  \textbf{40}(7),  475--488 (2001)

\bibitem{VazouSJVP2014}
Vazou, N., Seidel, E.L., Jhala, R., Vytiniotis, D., Peyton-Jones, S.:
  Refinement types for {Haskell}. In: Proc.\ of {ICFP}. pp. 269--282 (2014)

\bibitem{WadlerF2009}
Wadler, P., Findler, R.B.: Well-typed programs can't be blamed. In: Proc.\ of
  {ESOP}. pp. 1--16 (2009)

\bibitem{WilliamsMW2018}
Williams, J., Morris, J.G., Wadler, P.: The root cause of blame: Contracts for
  intersection and union types. Proc. ACM Program. Lang.  \textbf{2}(OOPSLA),
  134:1--134:29 (Oct 2018)

\bibitem{ZhuJ2013}
Zhu, H., Jagannathan, S.: Compositional and lightweight dependent type
  inference for {ML}. In: Proc.\ of {VMCAI}. pp. 295--314 (2013)

\end{thebibliography}

\NOTE{Up to 18 pages including bibliography}

\end{document}